\documentclass[letterpaper,10pt,conference]{ieeeconf}

\IEEEoverridecommandlockouts                              
\usepackage{graphicx}      
\usepackage{color}
\usepackage{amsmath,amsfonts,amssymb, bm, dsfont}
\usepackage{latexsym}
\usepackage{algpseudocode}
\usepackage{algorithmicx}
\usepackage{tikz}
\usepackage{caption}
\usepackage{subcaption}

\renewcommand{\l}{\left}
\renewcommand{\r}{\right}

\newcommand{\mc}{\mathcal}

\newcommand{\1}{\mathds{1}}
\newcommand{\R}{\mathbb{R}}

\def\beq{\begin{equation}}
\def\eeq{\end{equation}}
\def\ba{\begin{array}}
\def\ea{\end{array}}

\newtheorem{theorem}{Theorem}

\newtheorem{proposition}[theorem]{Proposition}
\newtheorem{lemma}[theorem]{Lemma}
\newtheorem{corollary}[theorem]{Corollary}
\newtheorem{example}{Example}
\newtheorem{remark}{Remark}
\newtheorem{assumption}{Assumption}

\graphicspath{{graphics/}}

\title{\LARGE \bf
Targeting interventions for displacement minimization in opinion dynamics}

\author{Luca Damonte, Giacomo Como, and Fabio Fagnani
\thanks{The authors are with the Department of Mathematical Sciences ``G.L.~Lagrange,'' Politecnico di Torino, Corso Duca degli Abruzzi 24, 10129, Torino, Italy.  Email: luca.damonte@polito.it, giacomo.como@polito.it, fabio.fagnani@polito.it. Giacomo Como is also with the Department of Automatic Control, Lund University, Box 118
SE-221 00 Lund, Sweden. }
}

\begin{document}

\maketitle
\thispagestyle{empty}
\pagestyle{empty}

\begin{abstract} Social influence is largely recognized as a key factor in opinion formation processes. 
Recently, the role of external forces in inducing opinion displacement and polarization in social networks has attracted significant attention. 
This is in particular motivated by the necessity to understand and possibly prevent interference phenomena during political campaigns and elections. 
In this paper, we formulate and solve a targeted intervention problem for opinion displacement minimization on a social network.
Specifically, we consider a min-max problem whereby a social planner (the defender) aims at selecting the optimal network intervention within her given budget constraint in order to minimize the opinion displacement in the system that an adversary (the attacker) is instead trying to maximize. 
Our results show that the optimal intervention of the defender has two regimes. 
For large enough budget, the optimal intervention of the social planner acts on all nodes proportionally to a new notion of network centrality. 
For lower budget values, such optimal intervention has a more delicate structure and is rather concentrated on a few target individuals.
\end{abstract}



\section{Introduction}
The possibility to interfere in the democratic life of a society so to change the outcome of a political campaign or of an election, to alter the balance of power on a social issue, to legitimate or delegitimate a position is a concrete unsettling reality. 
Automatic programs, often referred to as "bots", are increasingly used to manipulate debates in social networks and to spread fake news and have become more sophisticated and harder to distinguish from real people. 

Over the past fifteen years, a large body of literature  has focused on the impact of the social network structure on shaping the emergent opinion profile in a society. 
Phenomena such as \emph{consensus, polarization}, or \emph{persistent disagreement} have been studied in this context, as well as the role of targeted interventions in shifting individuals' opinions in a desired direction.  

A fundamental model used to investigate these phenomena is the French-DeGroot \cite{degroot1974reaching} linear averaging model and its subsequent extensions including nonlinearities in the updating laws (e.g., bounded confidence models \cite{lorenz2007continuous}), randomness in the interactions, exogenous inputs (e.g., stubborn nodes \cite{Como.Fagnani:2016}). 
A celebrated model that belongs to this last research line is due to Friedkin-Johnsen \cite{friedkin1990social}. 

Building on such opinion dynamics models, various works in the recent past have studied the interaction between the network structure and the effect of external influences and formulated targeting optimization problems. 
In \cite{yildiz2013binary} and  \cite{vassio2014message}, the problem of optimizing the position of a stubborn agent to maximally offset the effect of other influencing agents has been investigated. 
A corresponding adversarial problem formulated as a zero-sum game has been proposed and analyzed in \cite{dhamal2018optimal}. 

For the special setting of the Friedkin-Johnsen opinion dynamics model, the work \cite{gaitonde2020adversarial} studies an adversarial problem where the first player (the attacker) can manipulate the initial opinions of the agents, while the second player (the defender) applies a defense action weakening the action of the first one (mathematically increasing the cost of the attacker). 
Specifically, the authors consider a family of min-max problems for a variety of performance indices (e.g. disagreement, displacement), cost norms and budget constraints for which they propose algorithms and study performance bounds.

The complementary problem concerning the identification and detection of sources of exogenous influence in a social network from the analysis of opinion trajectories over time has also been studied \cite{ravazzi2021ergodic}. 

In this paper, for a general linear influence model encompassing the French-DeGroot averaging model with exogenous inputs, we analyze an adversarial model in the style of the one considered in \cite{gaitonde2020adversarial}. 
Attacker and defender pair their actions at the level of the exogenous inputs, the first one choosing the offset values to apply, the second one choosing the mitigating factors to counteract. 
The considered  performance measure is the absolute displacement from an equilibrium condition. 

The main contribution of this paper is an explicit recursive solution of the min-max problem that shows how the optimal solution for the defender is to invest its mitigation resources on all nodes if her available budget is sufficiently high, or on just a subset of selected nodes if the budget is low.  Emerging from our analysis is a novel network centrality measure on the set of inputs that indicates which are the most influential inputs for the defender to focus her intervention efforts.  
Our results apply, in particular, to the problem analyzed in \cite{gaitonde2020adversarial} for the special case of the absolute displacement.

The problem of targeting intervention ---namely understanding on which nodes in a network an external action has the maximum effect--- is a fundamental question in a variety of contexts encompassing opinion dynamics. 
In quadratic game theory, the work \cite{ballester2006who} coined the term \emph{key players} to indicate the players whose removal, in a quadratic game, would mostly affect the aggregate performance of the system. 
A different intervention aimed at changing the payoff function of prescribed players has been discussed, again in the context of quadratic games (and more general linear best reply games)  in \cite{galeotti2020targeting} and \cite{demange2017optimal}. 
In all these target intervention problems, the analysis leads to a specific concept of centrality: the most central nodes are those where it is more convenient to exert the action. 
The centralities obtained in the above referenced papers are variations of the popular Bonacich centrality. 
Adversarial settings have been instead proposed in \cite{chen2018quantifying} and \cite{damonte2020systemic}. 
The centrality proposed in this paper is rooted on the set of inputs acting on the network and is thus substantially different from the Bonacich centrality.
Optimal targeting problems have also been recently addressed in \cite{Hunter:2022} in the context of opinion dynamics with stubborn agents and in \cite{Como2021targeting} in the context of binary supermodular games. 

In the last subsection of this section we present some basic notation used throughout the paper. 
The rest of the paper is organized as follows. 
In Section \ref{sec:model} we introduce the model and the min-max optimization problem. 
A relevant application, concerning the Friedkin-Johnsen opinion model, is also introduced. 
In Section \ref{sec:results} we present our main results concerning the aforementioned problem. 
In Section \ref{sec:examples} we show some examples and simulations. In a conclusive section we present some possible future developments.

\textbf{Notation:} $\R_+$ and $\R_{++}$ indicate the set of nonnegative and positive real numbers, respectively. 
The transpose of matrix $M$ is denoted by $M'$. 
Given a vector $x$ in $\R^n$, we denote with the symbol $[x]$ the diagonal matrix such that $[x]_{ii}=x_i$.

\section{Model and problem statement}\label{sec:model}

We consider discrete-time linear time-invariant dynamics \beq\label{dynamic} x(t+1)=Ax(t)+Bu\,,\qquad t\ge 0\,,\eeq where $A$ in $\R_+^{n\times n}$ and $B$ in $\R_+^{n\times m}$ are nonnegative matrices, $u$ in $\R^m$ is a constant input vector, and the vector $x(t)$ in $\R^n$ represents the system state vector at time $t\ge0$. 
Let $\mc R=\{1,2,\ldots,n\}$ be a set of (regular) agents and $\mc S=\{n+1,n+2,\ldots,n+m\}$ a set of exogenous sources (sometimes also referred to as stubborn agents in the literature \cite{acemouglu2013opinion}). 
We may then interpret the matrix $$G=\left(\ba{rcl}A&B\\0&I\ea\right)$$ as the weighted adjacency matrix of a directed graph $\mc G=(\mc V,\mc E)$ with node set $\mc V=\mc R\cup\mc S$, and link set $\mc E=\{(i,j)\;|\; G_{ij}>0\}$. 
We notice that, in this graph, sources have no outgoing links (they are sinks in the graph terminology). 
A link from a regular agent $i$ to a source $k$ indicates a direct positive influence exerted by the exogenous source $k$ on agent $i$. 
A walk from a regular agent $i$ to a source $k$ indicates instead an indirect influence mediated by other regular agents.
The graph of regular agents, described by matrix $A$, admits self-loops.

Throughout the paper we shall work under the following assumption. 
\begin{assumption}\label{assumption:A-Schur}
The matrix $A$ is Schur stable, i.e., it has spectral radius $\rho(A)<1$. 
\end{assumption}\medskip
It is well known that Assumption \ref{assumption:A-Schur} implies that, for every initial state $x(0)$ in $\R^n$, the dynamics \eqref{dynamic} converge to the globally exponentially stable equilibrium point
\beq\label{equilibrium} x=Mu\,,\eeq 
where
\beq\label{def:M}M=(I-A)^{-1}B\,.\eeq
We shall assume that the input vector $u$ has the form \beq\label{omega-nu} u=[\nu]^{-1/2}\omega,\eeq
where $\omega$ in $\R^m$ is an exogenous input, while $\nu$ in $\R_{++}^m$ is a protection vector that aims at attenuating the effect of the sources.
We measure the system performance in terms of the quadratic displacement 
\beq\label{displacement} \Phi(\nu,\omega)=||x||^2=\omega'\l[\nu\r]^{-1/2}M'M\l[\nu\r]^{-1/2}\omega\,,\eeq
and consider the following min-max optimization problem:
\beq\label{minmax}\min\limits_{\substack{\nu\ge d\\\sum_i\nu_i\le c}}\max\limits_{\|\omega\|_2\le 1}\; \Phi(\nu,\omega)\,,\eeq
with $d$ in $\R_{++}^m$ and $c\ge\sum_id_i$. 
The interpretation is that of an adversarial model where two agents, an exogenous disturber and a system defender, compete. 
The former seeks to maximize the displacement from the  equilibrium point $0$ (that would be obtained if $\omega=0$), whereas the latter aims at reducing such displacement. 
The action of the disturber has been bounded by $1$ in Euclidean norm. 
This does not entail any loss of generality as a different bound can be absorbed in the bound of the defender. 
The defender has a limited budget $c$ and moreover has a lower bound $d$ on the single components. 
Typically, we will assume that $d_i\geq 1$ that means that the defender can not amplify the disturber action in any component.

\subsection{Example (Friedkin-Johnsen dynamics)}\label{sec:FriedkinJohnsen}
The Friedkin-Johnsen model \cite{friedkin1990social} in opinion dynamics can be described by $$x(t+1)=[\lambda] Px(t)+(I-[\lambda])u,$$ where $P$ is a stochastic nonnegative matrix of influences and $[\lambda]$ is a diagonal matrix such that $\lambda_{i}\in[0,1]$ for every $i$. 
In this context, input $u_i$ is referred to as the anchor of agent $i$ and is typically taken to be the original belief of the agent, namely $u=x(0)$. 
The vector $x(t)$ in $\R^n$ collects the opinions that a set of agents hold on a certain fact at time $t$. 

Clearly, this fits in the general model (\ref{dynamic}) with $A=[\lambda] P$ and $B=(I-[\lambda])$. 
If in the graph $\mc G$ every regular node has a walk to a source (this is automatic if $\lambda_{i}<1$ for every $i$), it can be shown that $A$ is Schur stable. 
For this specific model, the authors of \cite{gaitonde2020adversarial}  study the min-max problem 
\beq\label{minmax gaitonde} \min\limits_{\substack{\nu>0\\\sum_i\nu_i= n}}\;\max\limits_{\substack{u\in\R^n\\\|[u]\nu^{1/2}\|^2\leq 1}}\; u'M'Mu.\eeq
In this model, the action of the defender modifies the cost norm of the attacker.
The adversarial interpretation can be that the action of the defender makes certain nodes less vulnerable and, in order to produce the same result, the attacker needs to put on a stronger more costy effort. 
We notice that making the change of variable $\sqrt{\nu_i}u_i=\omega_i$ and setting $c=n$ and $d=0$, the above problem can be equivalently expressed in the form (\ref{minmax}).

In the next section we investigate the solution of such problems and study the nature of optimal weights, denoted $\nu^*$, that solves them. 
We highlight that the main goal of our study is to exactly characterize the optimal intervention of the defender and understand how the budget should be optimally allocated through agents. 
Therefore, our study is a natural extension of \cite{gaitonde2020adversarial}, where an algorithm to solve a much more general problem is given but no information on the optimal solution is investigated.

\section{Main results}\label{sec:results}
In this section, we present the most important results of the paper. 

First, notice that Assumption \ref{assumption:A-Schur} allows one to express the matrix $M$ defined in \eqref{def:M} as the limit of the matrix series 
\beq\label{Mneumann} M=\sum_{k=0}^{+\infty}A^kB\,.\eeq
From the above one can easily deduce that $M_{ij}\ge0$ and $M_{ij}>0$ if and only if the exogenous source $j$ in $\mc S$ is reachable from agent $i$ in $\mc V$ in the graph $\mc G$. 
It then follows that 
\beq\label{def:H} H=M'M\eeq
is a nonnegative symmetric matrix such that $H_{ij}>0$ if and only if there exists at least a regular agent $k$ from which both sources $i$ and $j$ are reachable in the graph $\mc G$. 
In fact, we may interpret $H$ as the weighted adjacency matrix of a new undirected graph $\mc H=(\mc S,\mc F)$ with node set $\mc S$, whereby $i$ and $j$ are linked by an undirected link of weight $H_{ij}$ if and only if there exists at least a regular agent $k$ that is (possibly indirectly) influenced by both sources $i$ and $j$ in the original network $\mc G$. 
We shall denote by 
\beq\label{centrality}\pi=\frac1{\sum_{i,j}H_{ij}}H\1\,,\eeq 
the normalized degree centrality vector of this graph.

We shall make the following assumption. 
\begin{assumption}\label{assumption:Hirreducible}
The matrix $H$ is irreducible. 
\end{assumption}\medskip

\begin{remark}
Assumption \ref{assumption:Hirreducible} is equivalent to requiring that the graph $\mc H$ is connected. 
Notice that a sufficient condition for Assumption \ref{assumption:Hirreducible} to be satisfied is that $A+A'$ is irreducible.
From a network connectivity point of view, irreducibility of the matrix $A+A'$ is equivalent to that the graph $\mc G$ be weakly connected, that is, replacing all directed edges with undirected ones generates a connected graph. However, notice that Assumption \ref{assumption:Hirreducible} is a weaker requirement than the graph $\mc G$ being weakly connected. 
\end{remark}

We first set some further notation. 
Let a positive vector $d$ in $\R^m_{++}$ be given. 
For every scalar $c\geq\1'd$, we define the set of admissible protection vectors as
\beq\label{Qc} \mc Q_{c}=\{\nu\in\R^m\;|\; \nu_i\geq d_i, \; \1'\nu\leq c\}.\eeq 
We approach the min-max problem (\ref{minmax}) by first analyzing the inner minimization problem. 
To this aim, we define the function $\phi: \mc Q_{c}\to\R$ as
$$\phi(\nu)=\max_{\|\omega\|_2=1}\; \Phi(\nu,\omega),\quad\nu\in \mc Q_{c}.$$ 
The next result gathers some important facts on the function $\phi$.
\begin{lemma}\label{lemma:delta}
Let $A$ in $\R_+^{n\times n}$ and $B$ in $\R_+^{n\times m}$ satisfy Assumption \ref{assumption:A-Schur}. 
Let $d$ in $\R_{++}^m$ and $c\ge\1'd$. 
Then, for every $\nu$ in $\mc Q_{c}$, we have 
\begin{enumerate}
\item[(i)] $\phi(\nu)=\rho\l([\nu]^{-1/2}H[\nu]^{-1/2}\r)=\rho\l(M\l[\nu\r]^{-1}M'\r)$. 
\end{enumerate}
Moreover, if also Assumption \ref{assumption:Hirreducible} is satisfied, then:
\begin{enumerate}
\item[(ii)] $\phi(\nu)$ is a simple eigenvalue of $[\nu]^{-1/2}H[\nu]^{-1/2}$
\item[(iii)] $\phi(\nu)$ is strictly convex in $\nu$.
\item[(iv)]
\beq\label{derivative} \frac{\partial}{\partial \nu_i}\phi(\nu)=-(M'z)_i^2/\nu_i^2\,,\eeq
where $z$ is the dominant eigenvector of $M\l[\nu\r]^{-1}M'$ associated to the eigenvalue $\phi(\nu)$.
\end{enumerate}
\end{lemma}
\medskip

Since $ \mc Q_{c}$ is convex, Lemma \ref{lemma:delta} implies that the minimum of $\phi$ on such a set is unique. 
We denote it by $\nu^*(c)$ to remember its dependence on the budget $c$.

We first solve the min-max problem relatively to the unconstrained case where we drop the lower bound conditions expressed through $d$. 
We define, for every positive scalar $c$,
\beq\label{Qc0} \mc Q^0_{c}=\{\nu\in\R^m\;|\; \nu_i>0, \; \1'\nu\leq c\},\eeq
and we consider 
\beq\label{minmax0}\min\limits_{\substack{\nu\in\mc Q^0_c}}\max\limits_{\|\omega\|_2\le 1}\; \Phi(\nu,\omega)=\min\limits_{\substack{\nu\in\mc Q^0_c}}\phi(\nu)\,.\eeq
From Lemma \ref{lemma:delta} we deduce that the minimum above is unique and we indicate it as $\nu^0(c)$. 
An explicit form of $\nu^0(c)$ is presented below.
\begin{proposition}\label{prop:unconst} 
For every $c>0$, we have that $$\nu^0(c)=c\pi,\quad \min\limits_{\substack{\nu\in\mc Q_c^0}}\phi(\nu)=\phi(c\pi)=\frac{\1'H\1}{c}.$$
\end{proposition}
\begin{proof}
Monotonicity properties of $\phi$ (see relations (\ref{derivative})) imply that, necessarily, $\1'\nu^0(c)=c$. 
Using the explicit expression for the derivative of the objective function in (\ref{derivative}) and classical Lagrangian multipliers techniques, we obtain the following equations:
\beq\label{KKT1} \left\{\begin{array}{l}-\nu_i^{-2}(M'z)_i^2+\mu=0\quad i\in\mc S\\ \1'\nu=c\\ M[\nu]^{-1}M'z=\rho z\end{array}\right.\eeq
where $\mu$ is the Lagrangian multiplier, $\rho=\phi(\nu)$ is the value function and $z$ is the positive dominant eigenvector of $M[\nu]^{-1}M'$. 
We proceed as follows. 
The first bulk of equations yield
\beq\label{deriv1} \mu^{1/2}\nu=M'z.\eeq

Substituting in the third equation, we obtain
\beq\label{deriv2} \mu^{1/2}M\1=\rho z.\eeq
As $\rho>0$ and $z>0$, we derive from (\ref{deriv1}) and (\ref{deriv2}) that
\beq\label{deriv3}\nu= \mu^{-1/2}M'z=\rho^{-1}M'M\1.\eeq
The fact that $\1'\nu=c$ yields the thesis.
\end{proof}

We now notice that the unconstrained solution $\nu^0(c)$ computed in Proposition \ref{prop:unconst} satisfies 
$$\nu^0(c)\in  \mc Q_{c}\;\Leftrightarrow\; c\geq c^0=\max\limits_{i=1}^m\frac{d_i}{\pi_i}.$$
For such values of $c$, in consideration of the fact that $\mc Q_{c}\subseteq \mc Q^0_c$, we have that the solution found also solves the constrained minimum problem. 
We gather this in the following result.

\begin{theorem}\label{theo:main-1}
Let $A\in\R_+^{n\times n}$ and $B\in\R_+^{n\times m}$ satisfy Assumptions \ref{assumption:A-Schur} and \ref{assumption:Hirreducible}. 
Then, for every $d\in\R_{++}^m$ and $ c\ge c^0$ we have 
\beq\label{min}\nu^*(c)=c\pi,\quad \phi(\nu^*(c))=\frac{\1'H\1}{c}.\eeq
\end{theorem}
\medskip

In the rest of the paper, we refer to the range $ c\ge c^0$ as to the \textit{high budget} regime. 
We recall one more time that the threshold value $c^0$  depends on the topology of the network, through $\pi$, but also on the lower bound vector $d$.

When $c$ instead satisfies $ \1'd\leq c< c^0$, we have that $\nu^0(c)\not\in  \mc Q_{c}$ and, consequently, the solution of the minimum constrained problem exhibits one or more components saturated to their lower bound level, namely $\nu_i^*(c)=d_i$ for some $i\in\mc S$. 
For later use, we define the set
$$\mc U_c=\{i\in\mc S\;|\; \nu^*_i(c)>d_i\}$$
that we call the set of active protections at budget level $c$. 

As a pivot to study this more general case, we investigate a  minimum problem where we assume that some of the variables $\nu_i$ are constrained to their lower bound level $d_i$, while the others are totally unconstrained.

We first set some notation. 
We assume agents to be split into two disjoint subsets $\mc U, \mc W$: $$\mc S=\mc U\cup\mc W.$$ 
Agents in $\mc U$ are unconstrained while agents in $\mc W$ are the constrained ones. 
We define the new set of variables
$$\mc Q^{\mc U}_{c}=\{\nu\in\R_+^m\;|\; \nu_i=d_i\;\forall i\in\mc W,\; \1'\nu\leq c\},$$
and we split, accordingly,
$$M=\l(\begin{array}{cc} M_{\mc U} & M_{\mc W}\end{array}\r).$$
We want to study 
\beq\label{minmax00}\min\limits_{\substack{\nu\in\mc Q^{\mc U}_{c}}}\phi(\nu)\,.\eeq

As $\mc Q^{\mc U}_{c}$ is still convex, Lemma \ref{lemma:delta} yields that the minimum above is unique. 
We denote it as $\nu^{\mc U}(c)$ and we explicitly derive it using first order conditions as in the proof of Proposition \ref{prop:unconst}. 

\begin{proposition}\label{prop:unconst2} 
Let $A\in\R_+^{n\times n}$ and $B\in\R_+^{n\times m}$ satisfy Assumptions \ref{assumption:A-Schur} and \ref{assumption:Hirreducible}. 
We consider a partition $\mc S=\mc U\cup\mc W$ where $\mc U\neq\emptyset$ and we fix $d\in \R^{\mc W}_{++}$. 
Then, for every $c\ge\1'd$, the solution $\nu^{\mc U}(c)$ is completely described by these relations:
\beq\label{nu^o}\left\{\begin{array}{l}\nu^{\mc U}(c)=(M_{\mc U}'\l(\rho I-M_{\mc W}[d]^{-1}M_{\mc W}'\r)^{-1}M_{\mc U}\1, d)\\
\phi(\nu^{\mc U}(c))=\rho\\
\1'M_{\mc U}'\l(\rho I-M_{\mc W}[d]^{-1}M_{\mc W}'\r)^{-1}M_{\mc U}\1=c-\1'd.
\end{array}\right.
\eeq
\end{proposition}
\medskip

The following is the key technical result of this paper.
\begin{theorem}\label{theo:main-2}
Let $A\in\R_+^{n\times n}$ and $B\in\R_+^{n\times m}$ satisfy Assumptions \ref{assumption:A-Schur} and \ref{assumption:Hirreducible}.
Fix $d\in\R_{++}^m$. 
Then, the function $\nu^*:\mc Q_c\to\R^m_{++}$ is continuous and (entrywise) non decreasing.
\end{theorem}

Using the above result, we can give a complete description of the optimal solution $\nu^*(c)$. 
This is the content of the following result that follows directly from Theorem \ref{theo:main-2}.
\begin{corollary}\label{cor:final}
Let $A\in\R_+^{n\times n}$ and $B\in\R_+^{n\times m}$ satisfy Assumptions \ref{assumption:A-Schur} and \ref{assumption:Hirreducible}.
Fix $d\in\R_{++}^m$. 
Then,  there exists a finite sequence of points $$c^0=\max\limits_{i=1}^m\frac{d_i}{\pi_i}>c^1>\cdots >c^s=\1'd$$ and subsets $\mc S\supsetneq\mc U^{0}\supsetneq \mc U^1\supsetneq\cdots\supsetneq\mc U^{s-1}$ such that
$$\mc U_c=\left\{\begin{array}{ll}\mc S\quad &\hbox{if}\; c>c^0\\ \mc U^k\; &\hbox{if}\; c^{k+1}<c\leq c^{k},\; k=0,\dots ,s-1\end{array}\right.$$
$$\nu^*(c)=\left\{\begin{array}{ll}\nu^0(c)\quad &\hbox{if}\; c>c^0\\ \nu^{\mc U^k}(c)\; &\hbox{if}\; c^{k+1}<c\leq c^{k},\; k=0,\dots ,s-1.\end{array}\right.$$
\end{corollary}
\medskip

This result says that the optimal protection vector $\nu^*$ exhibits a 'waterfiling' structure as a function of the budget $c$.
As we had already noticed, above the threshold $c^0$, $\nu^*$ coincides with the minimum of the unconstrained case. 
At $c=c^0$ one or more components saturate at their lower bound $d_i$ and for an interval of values $]c^1, c^0]$ all remaining components remain strictly above the corresponding lower bound and the optimum $\nu^*$ coincides with the solution of problem (\ref{minmax00}) with $\mc U=\mc U_{c^0}$. 
At $c^1$ more components will saturate and will remain stable in an interval $]c^2, c^1]$ and so on.  
We notice that the specific form of the solutions of problem (\ref{minmax00}) can be used to make Corollary \ref{cor:final} an effective recursive characterization of the optimal solution. 
Notice, in particular, that the sequence of thresholds $c^k$ and the subsets $\mc U^k$ can be recursively computed through the following characterization
\beq\label{cU_k}\begin{array}{rcl}c^{k+1}&=&\inf\{c\leq c^k\;|\; \nu^{\mc U^k}_i(c)>d_i\;\forall i\in \mc U^k\}\\
\mc U^{k+1}&=&\{i\in\mc S\;|\; \nu^{\mc U^k}_i(c^{k+1})>d_i\}.\end{array}\eeq

\begin{figure}[h]
\begin{algorithmic}[1]
\Procedure{Iterative solution}{$M, d$}
\State Set $k=0$,
\State $c^k=\max_{1\le i\le m}d_i/\pi_i,$ and $\mc U^k=\{i:c^k\pi_i>d_i\}$
\While{$|\mc U^k|>0$}
	\State Calculate $\phi(\nu^{\mc U^k})$ and $\nu^{\mc U^k}(c)$ through (\ref{nu^o})
	\State $k=k+1$
	\State Update $c^k$ and $\mc U^k$ using (\ref{cU_k})
	\State $\nu^*(c)=\nu^{\mc U^{k-1}}(c)$ for $c\in[c^{k},c^{k-1}]$
\EndWhile
\EndProcedure
\end{algorithmic}
\caption{Iterative procedure to calculate $\nu^*(c)$.}
\label{fig:algorithm}
\end{figure}

\section{Examples}\label{sec:examples}
In this section, we discuss three examples that illustrate our results and prove the effectiveness of our algorithmic solution.
We restrict to the case of a Friedkin-Johnsen model as described in Subsection \ref{sec:FriedkinJohnsen}. 
In all the examples, we assume that the regular nodes are connected through an undirected graph having $W$ as adjacency matrix and we define $P$ as $P_{ij}=w_i^{-1}W_{ij}$ where $w_i=\sum_{j\in\mc R}W_{ij}$ is the out-degree of node $i$. 
Moreover, we assume that $\lambda_{i}=1/2$ for every $i$. 
Notice that, in this case, we have that $M=(2I-P)^{-1}$. 
Finally, we set $d=\1$. 
Notice that the threshold for the high budget regime becomes $$ c^0=\frac1{\min_i\pi_i}$$ so that exclusively depends on the network topology trough the vector of centralities $\pi$.

\begin{example}\label{ex:regular}
Suppose that the graph connecting the regular nodes is undirected and $w$-regular (i.e. all nodes have the same degree $w_i=w$).
As $P$ is in this case symmetric, we also have that $M=(2I-P)^{-1}$ is symmetric, so that $H=M^2$. 
A direct check shows that $M\1=\1$ so that also $H\1=\1$. 
This implies that the centrality $\pi$ coincides with the uniform distribution, namely $$\pi=n^{-1}\1,$$ that is, all the agents have the same centrality.
Given that $d=\1$ it results that $c^0=n=\1'd$. 
This means that, in this case, there is only the high budget regime and the optimal intervention $\nu^*(c)$ is given by $\nu_i^*(c)=c/n$ for every $c\geq n$  and for every $i$.
\end{example}
\medskip

\begin{example}\label{ex:small}
Consider the undirected network of $n=11$ agents in Figure \ref{fig:network}. 
The three columns of Table \ref{tb:importance} represent, respectively, the index of nodes, their centralities $\pi_i$, and their degrees $w_i$. 

\begin{figure}[!ht]
\begin{subfigure}{0.2\textwidth}
\centering
\includegraphics[width=3.5cm]{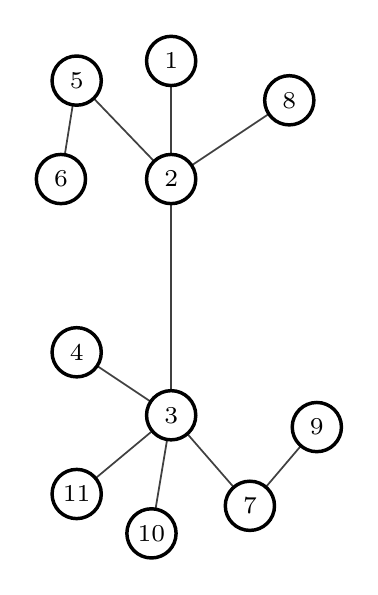}
\caption{Undirected network of $n=11$ nodes.}
\label{fig:network}
\end{subfigure}
\hfill
\begin{subfigure}{0.2\textwidth}
\centering
\begin{tabular}{|l|l|l|} 
\hline
$i$ & $\pi$ & $w$ \\ 
\hline\hline
 1  & 0.0647 & 1  \\
 2  & 0.1536 & 4  \\
 3  & 0.1854 & 5 \\
 4  & 0.064 & 1 \\
 5  & 0.0999 & 2 \\
 6  & 0.0704 & 1 \\
 7  & 0.0991 & 2 \\
 8  & 0.0647 & 1 \\
 9  & 0.0702 & 1 \\
10  & 0.064 & 1 \\
11  & 0.064 & 1 \\
\hline
\end{tabular}
\caption{Centrality $\pi$ and degree $w$ of the network.}
\label{tb:importance}
\end{subfigure}
\caption{Undirected network of $n=11$ agents and table of centrality $\pi$ and degree $w$ of each agent.}
\end{figure}

Figure \ref{fig:nu} shows the behaviour of the various components of $\nu^*(c)$ for $c$ in $[n,2n]$ with, in evidence, the threshold points where one or more components saturate at their lower level $1$. 
In Figure \ref{fig:phi-perf} we present a performance comparison between the optimal intervention and two heuristics based on the degree. 
In particular, we denote by $\nu^w(c)$ the intervention that distributes the budget $c$ among agents proportionally to the degree, i.e. $\nu_i^w(c)=1+(c-n)w_i/\sum_jw_j$. 
We denote by $\nu^{key}(c)$ the intervention that concentrates all the available budget on the node with the highest degree, that, in this case, is node $3$. 

\begin{figure}[!ht]
\begin{subfigure}{0.5\textwidth}
\centering
\includegraphics[width=8.5cm]{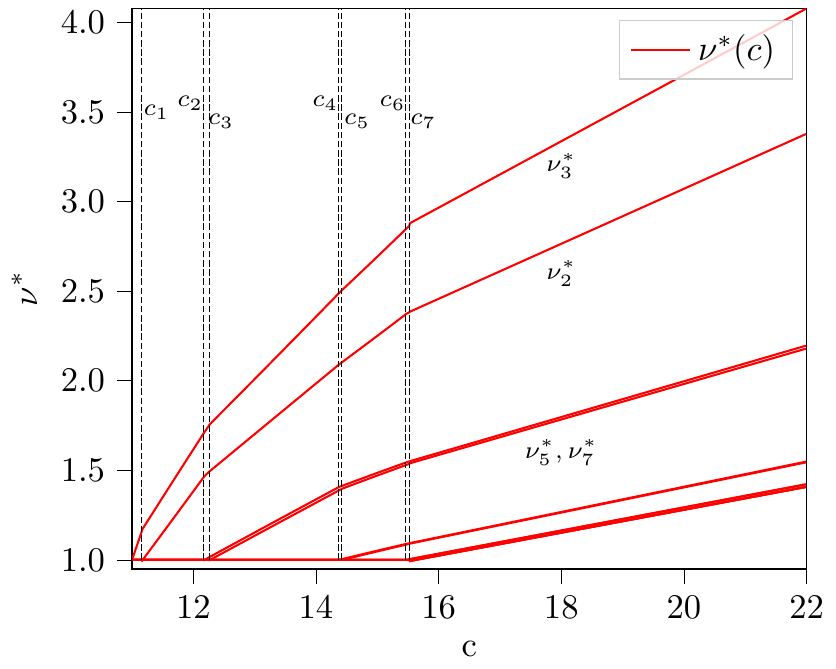}
\caption{Optimal intervention $\nu^*(c)$ for network depicted in Figure \ref{fig:network}.}
\label{fig:nu}
\end{subfigure}
\vfill
\begin{subfigure}{0.5\textwidth}
\centering
\includegraphics[width=8.5cm]{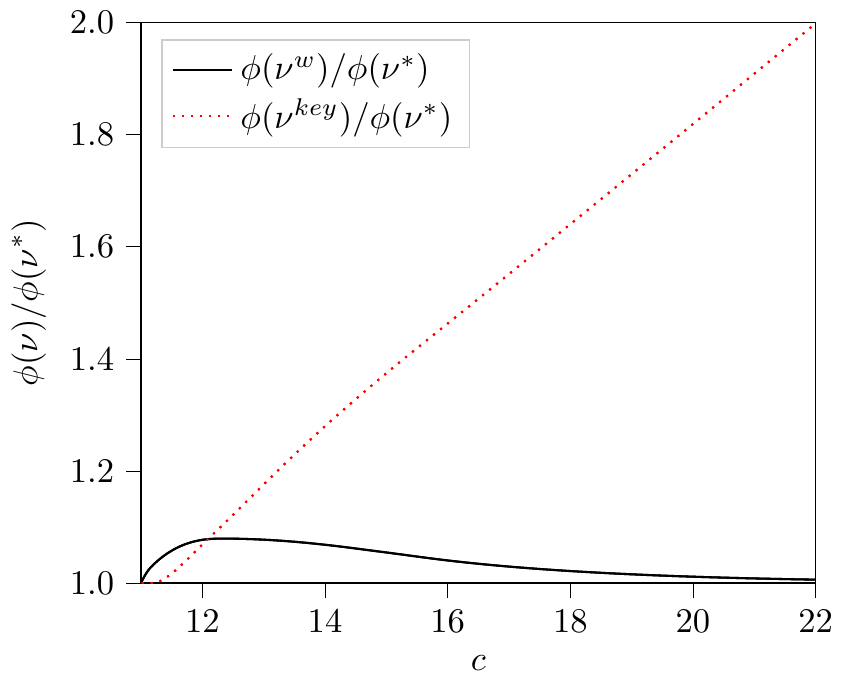}
\caption{Performance comparison between $\nu^*(c)$ and the two heuristics $\nu^w(c)$ (solid black) and $\nu^{key}(c)$ (dotted red).}
\label{fig:phi-perf}
\end{subfigure}
\caption{Behavior of optimal intervention $\nu^*(c)$ for $c\in[n,2n]$ and comparison between $\phi(\nu^w(c))/\phi(\nu^*(c))$ and $\phi(\nu^{key}(c))/\phi(\nu^*(c))$.}
\end{figure}

\end{example}
\medskip

\begin{example}\label{ex:topologies}
In this last example, we compare the behavior of the optimal value $\phi(\nu^*(c))$ and the range of the high budget regime with respect to different topologies. 
Specifically, we consider three networks with $n=20$ agents: a regular network, an Erd\H{o}s-Rényi (ER) undirected network, and a power-law (BA) preferential attachment network.
In the ER network we take a probability $p=1/4$ of connection. 
The BA network is constructed adding one new node at a time.

\begin{figure}[!ht]
\begin{center}
\includegraphics[width=8.5cm]{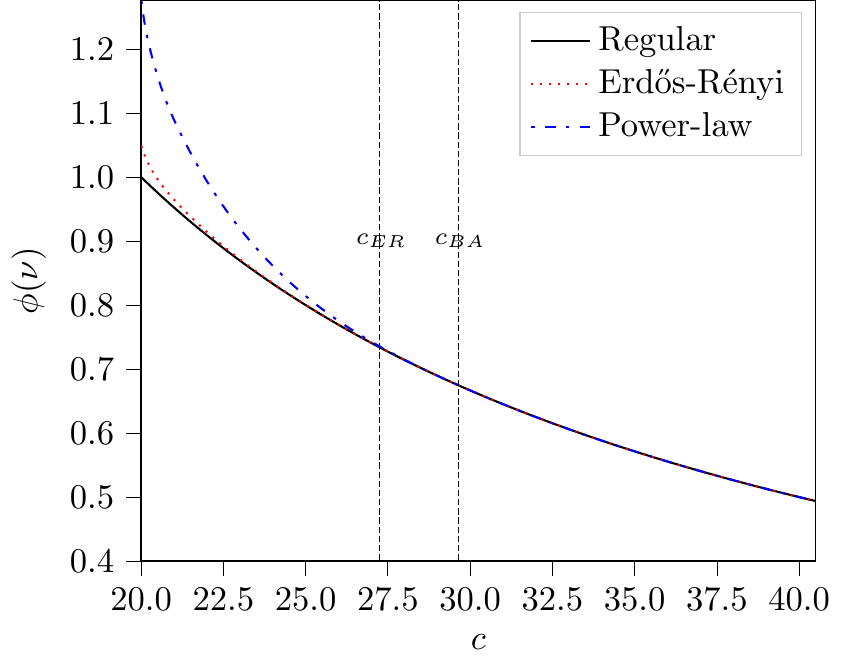}
\end{center}
\caption{Optimal value $\phi(\nu^*(c))$ for three different networks of $n=20$ agents: regular (solid black), Erd\H{o}s-Rényi with probability $p=1/4$ (dotted red), and power-law constructed with preferential attachment of one node at a time (dash-dotted blue). D
ashed vertical lines represent high budget regime $c^0$ for Erd\H{o}s-Rényi ($c_{ER}$) and power-law ($c_{BA}$) network.}
\label{fig:c high}
\end{figure}

In Figure \ref{fig:c high} are depicted values of $\phi(\nu^*(c))$ for the three networks considered. 
We notice two facts. 
First, the high budget threshold of the BA network is greater then the high budget threshold of ER network. 
This is due to the higher heterogeneity of the network centrality $\pi$ of the BA model due to the presence of hubs. 
Second, in the low budget regime, it is more costly to protect the BA model than the ER or the regular model.
\end{example}

\section{Conclusion}
We have presented and discussed an adversarial min-max problem for opinion displacement in a linear opinion dynamics model with exogenous inputs. 
This model encompasses the classical French-DeGroot model with stubborn agents and leads to a new concept of centrality rooted on the set of inputs. 
Our main result is an effective exact recursive solution of the min-max problem. 
A key property is the existence of a threshold value for the budget of the defender, above which the optimal intervention for the defender is an action on all sources proportional to their centrality, while below it, it concentrates on a subset of them.

Several directions are open for future research. 
Significant extensions are in the direction of removing Assumption 2 or considering more general performance functions (e.g.,  polarization). 
Another direction is to consider scenarios where the attacker employs random perturbation strategies.

\section{ACKNOWLEDGMENTS}
This work was partly supported by the Italian Ministry for University and Research through grants ``Dipartimenti
di Eccellenza 2018–2022'' [CUP: E11G18000350001] and Project PRIN 2017 ``Advanced Network Control of Future Smart Grids''
(http://vectors.dieti.unina.it), and by the {\it Fondazione Compagnia di San Paolo} through Project ``SMaILE''.

\bibliographystyle{IEEEtran}
\bibliography{cdc2022.bib}

\subsection{Proofs}
\begin{proof}[Proof of Lemma  \ref{lemma:delta}] Consider $\nu$ in $\mc Q_{c}$.

(i) The first identity follows from the  variational characterization (Courant–Fischer–Weyl min-max principle \cite[Theorem 4.2.6]{horn2012matrix}) of the eigenvalues of the symmetric matrix $[\nu]^{-1/2}H[\nu]^{-1/2}$. 
The second one follows from the property that non-zero eigenvalues of the product of two matrices are invariant with respect the order of the factors. 

(ii) It follows from Assumption \ref{assumption:Hirreducible} that $[\nu]^{-1/2}H[\nu]^{-1/2}$ is irreducible. 
Since $[\nu]^{-1/2}H[\nu]^{-1/2}$ is a nonnegative matrix, its dominant eigenvalue is simple \cite[Theorem 1.4]{berman1994nonnegative}. 

(iii) To prove the third item consider the function $f: \R^n_+\to \R$ given by \beq\label{f} f(s):=\rho\l(M[s]M'\r).\eeq 
Since we can express $$f(s)=\max\limits_{y\in\R^n:\|y\|_2=1}\|M[s]M'y\|_2$$ it follows that $f(s)$, being the max of convex functions, is convex. 
The same expression yields, since $M$ has all non-negative elements and the maximum is reached by a non-negative vector $k$, that $f(s)$ is non decreasing with respect to the component-wise order in $\R^n_+$.
 
We now prove convexity of $\phi(\nu)$. 
Notice that we can write $\phi(\nu)=f(\nu^{-1})$ where we indicate with $\nu^{-1}$ the vector of component-wise inversion of the vector $\nu$. 
Fix now $\nu_1,\nu_2 \in\R^n_+$ and $\lambda \in [0,1]$ and notice that, since the inversion function $\nu\mapsto \nu^{-1}$ is convex on $\R_+$, 
\beq\label{-1conv} (\lambda\nu_1+(1-\lambda)\nu_2)^{-1}\leq \lambda\nu_1^{-1}+(1-\lambda)\nu_2^{-1}\eeq component-wise. 
We now compute as follows
$$\begin{array}{rcl} \phi(\lambda\nu_1+(1-\lambda)\nu_2) &=& f((\lambda\nu_1+(1-\lambda)\nu_2)^{-1})\\ &\leq& f(\lambda\nu_1^{-1}+(1-\lambda)\nu_2^{-1})\\ &\leq& \lambda f(\nu_1^{-1})+(1-\lambda)f(\nu_2^{-1})\\ &=& \lambda\phi(\nu_1)+(1-\lambda)\phi(\nu_2),\end{array}$$
where in the first inequality we have used (\ref{-1conv}) and the monotonicity of $f$ and in the second inequality the convexity of $f$. 
Strict convexity follows from the fact that for every $i=1,\dots,n$, there exists $h$ such that $M_{ih}>0$; then the function $f$ is strictly increasing and thus the first inequality is strict.

(iv) follows by applying Implicit function theorem and explicit computation could be found in (\cite{magnus2019matrix}).
\end{proof}

\begin{proof}[Proof of Proposition \ref{prop:unconst2}]
Proof is analogous to the one of Proposition \ref{prop:unconst}. 
We notice that also in this case, the optimum will lay on the boundary $\1'\nu=c$. 
We write $(\nu=\tilde\nu, d)$ and we similarly derive first order conditions.
\beq\label{KKT2} \left\{\begin{array}{l}-\tilde\nu_i^{-2}(M'z)_i^2+\mu=0\quad i\in\mc U\\ \1'\tilde\nu=c-\1'd\\ M_{\mc U}[\tilde\nu]^{-1}M_{\mc U}'z+M_{\mc W}[d]^{-1}M_{\mc W}'z=\rho. z\end{array}\right.\eeq
The first set of equations yield
\beq\label{deriv11} \mu^{1/2}\tilde\nu=M_{\mc U}'z.\eeq
Substituting in the third equation, we obtain
\beq\label{deriv21} \mu^{1/2}M_{\mc U}\1+M_{\mc W}[d]^{-1}M_{\mc W}'z=\rho z.\eeq
Notice now that, since $\mc W\neq\emptyset$, $(M_{\mc W}[d]^{-1}M_{\mc W}')_{ij}$ $\leq(M\l[\nu\r]^{-1}M')_{ij}$ for all $i$ and $j$ with strict inequality on the diagonal terms. 
This implies that $\rho(M_{\mc W}[d]^{-1}M_{\mc W}')<\rho(M\l[\nu\r]^{-1}M')=\rho$ (for instance representing the spectral radius as norm matrix, details can be found in \cite[Theorem 8.1.18]{horn2012matrix}. 
Consequently,
\beq\label{deriv22} z=\mu^{1/2}(\rho I- M_{\mc W}[d]^{-1}M_{\mc W}')^{-1}M_{\mc U}\1.\eeq
Relation (\ref{deriv22}) together with (\ref{deriv11}) yield the thesis.
\end{proof}

Before proving Theorem \ref{theo:main-2}, we state and proof a simple lemma.

\begin{lemma}\label{lemma:min-high} 
Consider the solution $\nu^{\mc U}(c)$ of problem (\ref{minmax00}). 
For every $i\in\mc U$, $\nu_i^{\mc U}(c)$  is strictly increasing in $c$.
\end{lemma}
\begin{proof} 
Put $\nu^{\mc U}(c)=(\tilde\nu, d)$. 
From the first relation of (\ref{nu^o}), we can write $$\tilde\nu=\sum_{j=0}^{\infty} \frac{1}{(\phi(\nu^{\mc U}(c)))^{j+1}}\,M_{\mc U}'\l(M_{W}[d]^{-1}M_{\mc W}'\r)^jM_{\mc U}\1\,.$$ 
Since the value function $\phi(\nu^{\mc U}(c))$ is strictly decreasing in $c$, the result is proven.
\end{proof}

\begin{proof}[Proof of Theorem  \ref{theo:main-2}]
We start with continuity. 
Fix $c\geq \|d\|_1$ and consider a sequence $c_k\to c$ (with $c_k\geq \|d\|_1$) such that $\nu^*(c_k)\to\bar \nu$ for $k\to +\infty$. 
Consider now any sequence $\bar \nu_k\in\mc Q_{c_k}$ such that $\bar \nu_k\to \nu^*(c)$ for $k\to +\infty$. 
Since $\phi(\nu)$ is continuous, we have that $\phi(\nu^*(c_k))\to \phi(\bar \nu)$ and $\phi(\bar \nu_k)\to \phi(\nu^*(c))$.
Since by construction $\phi(\nu^*(c_k))\leq \phi(\bar \nu_k)$ for every $k$, it holds $\phi(\bar \nu)\leq \phi(\nu^*(c))$ and thus $\phi(\bar \nu)=\phi(\nu^*(c))$. 
Since $\phi(\nu)$ is strictly convex, it follows that $\bar \nu=\nu^*(c)$. 
This proves continuity.

We then prove monotonicity. 
Given any $\bar c\geq \|d\|_1$, we now show that $\mc W_c=\mc S\backslash\mc U_c$ is locally constant in a left and, respectively, in a right neighborhood of $\bar c$. 
To this aim, we consider
$$s_{\bar c}^-=\liminf\limits_{c \to \bar c-}|\mc W_{c}|\quad s_{\bar c}^+=\liminf\limits_{c\to \bar c+}|\mc W_{c}|.$$
Since $s_{\bar c}^-$ and $s_{\bar c}^+$ are integer-valued, there exists $\delta>0$ such that
\beq\label{inf1}\begin{array}{ll}
 |\mc W_{c}|\geq s_{\bar c}^-\; &\forall c\in [\bar c-\delta, \bar c[\\[3pt]
 |\mc W_{c}|\geq s_{\bar c}^+\; &\forall c\in ]\bar c,\bar c+\delta].
\end{array}\eeq
Consider any  $c_1\in [\bar c-\delta, \bar c[$ such that $|\mc W_{c_1}|=s_{\bar c}^-$. 
Since $\nu^*(c)$ is continuous, there exists $c_2>c_1$ such that $\mc W_{c}\subseteq \mc W_{c_1}$ for every $c\in [c_1, c_2[$.
Suppose we have chosen the supremum of such $c_2\leq\bar c$. 
By the way $c_1$ was chosen and the first inequality in (\ref{inf1}), we actually have that $\mc W_{c}= \mc W_{c_1}$ for $c\in [c_1, c_2[$. 
Since $\mc W_c$ is constant on $[c_1, c_2[$, it follows by continuity that $\mc W_{c_2}\supseteq \mc W_{c_1}$.
Applying monotonicity of Lemma \ref{lemma:min-high} we thus have $\mc W_{c_2}= \mc W_{c_1}$. 
If $c_2<\bar c$, repeating the same argument, we could further extend the interval $[c_1, c_2[$ on the right where $\mc W_c$ remains constant in contradiction with the way $c_2$ was chosen. 
Therefore $c_2=\bar c$ and we have proven that $\mc W_c$ is constant on $[c_1,\bar c]$. 
We now consider the right neighborhood. 
We fix any  $c_1\in ]\bar c, \bar c+\delta [$ such that $|\mc W_{c_1}|=s_{\bar c}^+$ and arguing as above we determine an interval $[c_1, c_2]$ on which $\mc W_c$ remains constant. 
By the definition of $s_{\bar c}^+$, we can fix $c_1$ arbitrarily close to $\bar c$ and this proves that $\mc W_c$ is constant on an interval $]\bar c, c_2]$. 
Lemma \ref{lemma:min-high} guarantees that on the two intervals $[c_1, \bar c]$ and $]\bar c, c_2]$ the optimal solution $\nu^*(c)$ is non-decreasing. 
Being continuous, it is non-decreasing in the neighborhood $[c_1, c_2]$ of $\bar c$. 
Finally, being locally non-decreasing and continuous, $\nu^*(c)$ is globally non-decreasing on $[\|d\|_1, +\infty[$. 
\end{proof}

\end{document}